\documentclass[runningheads]{llncs}
\usepackage[T1]{fontenc}
\usepackage{graphicx}
\usepackage{amsmath,amssymb}
\usepackage{algorithm}
\usepackage[noend]{algpseudocode}
\usepackage{array}
\usepackage{xcolor}
\usepackage{orcidlink}
\usepackage{subcaption}
\newcommand{\share}[1]{[\![#1]\!]}
\DeclareMathOperator{\Gen}{Gen}
\DeclareMathOperator{\Eval}{Eval}
\DeclareMathOperator{\Encode}{Encode}
\DeclareMathOperator{\Decode}{Decode}

\DeclareMathOperator{\DPF}{DPF}
\DeclareMathOperator{\DCF}{DCF}
\DeclareMathOperator{\DDH}{DDH}
\DeclareMathOperator{\Real}{Real}
\DeclareMathOperator{\Sim}{Sim}
\DeclareMathOperator{\Leak}{Leak}
\DeclareMathOperator{\Ideal}{Ideal}
\DeclareMathOperator{\Output}{Output}

\begin{document}
\title{DDH-based schemes for multi-party\\Function Secret Sharing}
\titlerunning{DDH-based schemes for multi-party FSS}
%
\author{Marc Damie\inst{1,2}\orcidlink{0000-0002-9484-4460}\thanks{Corresponding author: \email{m.f.d.damie@utwente.nl}}\and
      Florian Hahn\inst{1}\orcidlink{0000-0003-4049-5354}\and
      Andreas Peter\inst{3}\orcidlink{0000-0003-2929-5001}\and
      Jan Ramon\inst{2}}
\authorrunning{M. Damie et al.}
%
\institute{University of Twente, The Netherlands\and
Inria, France \and
Carl von Ossietzky Universität Oldenburg, Germany}
\maketitle              
\begin{abstract}
  Function Secret Sharing (FSS) schemes enable sharing efficiently secret functions.
  Schemes dedicated to point functions, referred to as Distributed Point Functions (DPFs), are the center of FSS literature thanks to their numerous applications including private information retrieval, anonymous communications, and machine learning.
  While two-party DPFs benefit from schemes with logarithmic key sizes, multi-party DPFs have seen limited advancements: $O(\sqrt{N})$ key sizes (with $N$, the function domain size) and/or exponential factors in the key size.

  We propose a DDH-based technique reducing the key size of existing multi-party schemes.
  In particular, we build an honest-majority DPF with $O(\sqrt[3]{N})$ key size.
  Our benchmark highlights key sizes up to $10\times$ smaller (on realistic problem sizes) than state-of-the-art schemes.
  Finally, we extend our technique to schemes supporting comparison functions.
\end{abstract}

\section{Introduction}
Secret sharing \cite{shamir_how_1979} is a popular cryptographic primitive to perform multi-party computations (MPC) \cite{evans_pragmatic_2018}.
This primitive enables splitting a secret value into several shares, revealing individually no information about the secret.
While classic secret sharing focused on scalar secret values, Function Secret Sharing (FSS) \cite{gilboa_distributed_2014} generalizes the concept to share secret functions.
FSS schemes consists in three algorithms: $\Gen$, $\Eval$, and $\Decode$.
$\Gen$ outputs $p$ FSS keys based on a secret function $f$; each shareholder receives one key.
$\Eval$ takes as input an FSS key $k_i$ and a point $x$, and outputs a share of $f(x)$ (referred to as $\share{f(x)}_i$).
Finally, $\Decode$ takes as input $p$ shares $\{\share{f(x)}_1, \dots,\share{f(x)}_p\}$  and outputs $f(x)$.

Significant efforts \cite{boyle_function_2015,boyle_function_2016,bunn_cnf-fss_2022,corrigan-gibbs_riposte_2015,gilboa_distributed_2014} aimed to share point functions (i.e., a function equal to zero everywhere, except on a point $\alpha$).
These so-called ``Distributed Point Functions'' (DPFs) were initially proposed to build private information retrieval (PIR) protocols \cite{gilboa_distributed_2014} and have been later used in other domains such as anonymous communications \cite{corrigan-gibbs_riposte_2015} or privacy-preserving machine learning \cite{wagh_pika_2022}.

Recently, there is also a growing interest in schemes supporting comparison functions: functions such that $f(x) = \beta$ when $x\le \alpha$, $0$ otherwise.
These schemes called ``Distributed Comparison Functions'' (DCF) have applications notably in MPC pre-computations \cite{boyle_secure_2019} or private statistics \cite{barczewski_differentially_2025}.

\paragraph{Related works}
The main goal of FSS works is to minimize the key size.
In particular, existing works systematically studied the influence of the function domain size $N$ on the key size.
While schemes are often described for a generic group $\mathbb{G}$, many works focused on bit-string outputs \cite{boyle_function_2015,bunn_cnf-fss_2022}: $\mathbb{G} = (\mathbb{F}_2)^l$; motivated by applications in PIR.
However, novel applications in private statistics \cite{barczewski_differentially_2025,boneh_lightweight_2021} or MPC precomputations \cite{boyle_secure_2019} require prime fields $\mathbb{F}_q$ (e.g., to sum several shared functions).
Such constraints emphasized under-studied scalability issues in some multi-party schemes \cite{boyle_function_2015,bunn_cnf-fss_2022} that have an exponential dependency in $q$.
Recent works \cite{boyle_function_2022,kumar_compact_2024} also highlighted this issue.

\paragraph{2/3-party DPF}
Boyle et al. \cite{boyle_function_2015} presented the reference 2-party scheme, with an $O(\log N \log q)$ key size.
This scheme has been further optimized by various works: key size optimization \cite{boyle_function_2016}, ``incremental'' DPF \cite{boneh_lightweight_2021}, reusable keys \cite{boyle_programmable_2022}, MPC-based key generation \cite{doerner_scaling_2017}, verifiable DPF \cite{de_castro_lightweight_2022}.
Thanks to its logarithmic key size, this scheme is used in many application papers including in anonymous communications \cite{newman_spectrum_2022}, private statistics \cite{boneh_lightweight_2021}, and access control \cite{servan-schreiber_private_2023}.

Bunn et al. \cite{bunn_efficient_2020} described a three-party scheme with $O(\sqrt{N})$ key size.
Later, \cite{bunn_cnf-fss_2022,zyskind_high-throughput_2024} further improved it to obtain $O(\log N)$ key sizes.

\paragraph{Multi-party DPF} Now that 2-party and 3-party schemes have logarithmic key sizes, improving arbitrary $p$-party schemes is the natural next step.
As we focus on multi-party FSS, we will skim through existing schemes, their advantages and disadvantages.
Table \ref{tab:dpf} of Section \ref{sec:dpf} summarizes this overview.

First, Boyle et al. \cite{boyle_function_2015} designed a scheme based on Pseudo-Random Generators (PRG) with $O(\sqrt{N}q^{\frac{p-1}{2}}\log q)$ key size.
The exponential factor $q^{\frac{p-1}{2}}$ (with $q = |\mathbb{G}|$) makes it impractical for applications requiring arbitrary output groups (e.g., private histograms \cite{boneh_lightweight_2021}).
This weakness was already highlighted by \cite{kumar_compact_2024}.

Second, Corrigan-Gibbs et al. \cite{corrigan-gibbs_riposte_2015} introduced a DDH-based scheme with $O(\sqrt{N}\log q)$ key size.
Recently, Kumar et al. \cite{kumar_compact_2024} further improved this scheme by removing a constant factor, but requires a \emph{trusted share decoder}.

Third, Bunn et al. \cite{bunn_cnf-fss_2022} designed an honest-majority scheme (based on \cite{boyle_function_2015}) with $O(\sqrt[4]{N}q^{\frac{p-1}{2}}\log q)$ key size.
It has the best dependency on $N$ (i.e., $O(\sqrt[4]{N})$), but it inherits the exponential factor of \cite{boyle_function_2015}.
Papers presenting PRG-based schemes \cite{boyle_function_2015,bunn_cnf-fss_2022} only provide algorithms for $q=2$, but they can easily be generalized to an arbitrary $q$ (as explained in \cite{boyle_function_2022}).

Fourth, Bunn et al. \cite{bunn_cnf-fss_2022} also introduced an honest-majority scheme with $O(\sqrt{N}\log q)$ key size.
The main advantage is \textbf{its information-theoretic security}.
Other works also investigated information-theoretic DPF schemes \cite{boyle_information-theoretic_2022,kruglik_verifiable_2024,li_efficient_2023}, but \cite{bunn_cnf-fss_2022} is the only with practical key sizes.

Finally, concurrently to our work, two other papers \cite{goel_multiparty_2025,krips_multi-party_2025} have improved multiparty DPF.
On the one hand, Goel et al. \cite{goel_multiparty_2025} built a PRG-based scheme replacing the exponential factor present in \cite{boyle_function_2015} with a polynomial factor: $O(\sqrt{N}\cdot {p^3\lambda^4})$.
On the other hand, Krips and Pullonen-Raudvere \cite{krips_multi-party_2025} introduced two new hardness assumptions to build a multi-party scheme with logarithmic key size.

\paragraph{DCF} Boyle et al. \cite{boyle_function_2015} adapted their 2-party DPF to build a DCF scheme with $O(\log N \log q)$ key size.
They also adapted their multi-party DPF to build a multi-party DCF with $O(\sqrt{N}q^{\frac{p-1}{2}}\log q)$; suffering from the same scalability issues as their DPF.
Recently, Kumar et al. \cite{kumar_compact_2024} proposed multi-party DCF inspired by the DPF of \cite{corrigan-gibbs_riposte_2015}.
The literature describes no other DCF; the DCF baselines are then weaker than those in DPF.

\paragraph{Gap in the literature} When relying on standard hardness assumptions (which then excludes \cite{krips_multi-party_2025}), the existing literature leaves two choices for multi-party DPF: either an $O(\sqrt[4]{N})$ key size with exponential factors \cite{bunn_cnf-fss_2022}, or an $O(\sqrt{N})$ key size without exponential factors \cite{bunn_cnf-fss_2022,corrigan-gibbs_riposte_2015,damie_eliminating_2025}.

\textbf{Our goal is to provide an intermediary solution, a better trade-off}: $O(\sqrt[3]{N})$ key size without exponential factors.
Similarly, we want to extend these results to DCF, and provide the same choice range as in DPF.

We would like to emphasize the contribution of \cite{krips_multi-party_2025} which, for the first time, presents a multi-party scheme with logarithmic key size.
Unfortunately, it requires the introduction of two new hardness assumptions.
Further research is necessary to estimate the parameters under which these assumptions hold.
Our paper, on the other hand, relies on standard, widely-established assumptions.

\paragraph{Our Contributions}
\begin{enumerate}
  \item We build a \textbf{multi-party DDH-based DPF with $O(\sqrt[3]{N})$ key size}. Our scheme provides keys \textbf{up to $\times$10 smaller} than existing works.
  \item We \textbf{extend our approach to build a DCF} with $O(\sqrt[3]{N})$ key size.
  \item As DDH-based schemes require encoding the secret as a DDH group element, we present \textbf{two encodings, discuss their properties and applications}.
\end{enumerate}

\section{Definitions}
Let $p$ be the number of parties/shareholders, $m$ be the number of dishonest parties.
Let $\mathbb{F}_q$ be a prime field and let $\mathbb{G}$ be a cyclic group.
Let $N$ be the function domain size and $1^\lambda$ a security parameter.
Finally, let $s \xleftarrow{R} S$ be the uniform sampling of an element $s$ from the set $S$.
Let $\share{x}$ be a share of $x$, and $g^{\share{x}}$ to $g$ to the power $\share{x}$. 

\subsection{Threat model}
MPC protocols are referred to as secure if they preserve output correctness and input privacy in presence of an adversary.
Subsection \ref{subsec:def-fss} provides definitions of correctness and privacy specific to FSS.

Like most FSS works \cite{boyle_function_2015,boyle_function_2016,bunn_cnf-fss_2022,gilboa_distributed_2014,kumar_compact_2024}, we focus on semi-honest adversaries; adversaries following the protocol and \emph{passively} infering secret information.

Moreover, secure protocols are characterized by the tuple $(m,p)$.
The security of an $(m,p)$-secure protocol is guaranteed only if the number of dishonest parties is below than or equal to $m$.
The literature distinguishes honest-majority protocols ($m<p/2$) from dishonest-majority protocols ($m\ge p/2$).

\subsection{Function Secret Sharing}
\label{subsec:def-fss}
Function secret sharing (FSS) \cite{boyle_function_2015} generalizes the concept of secret sharing to secret functions.
Each FSS scheme can share functions from a specific function family.
A function family $\mathcal{F}$ \cite{boyle_function_2022} is a pair $(P_\mathcal{F}, E_\mathcal{F})$ where $P_\mathcal{F} \subseteq \{0,1\}^*$ is an infinite collection of function descriptions $\hat{f}$, and $E_\mathcal{F}$: $P_\mathcal{F} \times \{0, 1\}^* \rightarrow \{0, 1\}^*$ is a polynomial-time algorithm defining the function described by $\hat{f}$.

In other words, each function description $\hat{f} \in P_\mathcal{F}$ describes a corresponding function $f$ such that $f(x) = E_\mathcal{F}(\hat{f}, x)$.
A description $\hat{f}$ for such functions is the tuple $(\alpha, \beta, \mathcal{X}, \mathcal{Y})$, with $\mathcal{X}$ the function domain and $\mathcal{Y}$ the output space.

Over the years, the literature has described schemes for various function families especially point functions \cite{boyle_function_2015,boyle_function_2016,boyle_information-theoretic_2022,bunn_cnf-fss_2022,corrigan-gibbs_riposte_2015,gilboa_distributed_2014} (functions $f$ such that $f(x) = \beta$ if $x=\alpha$, and $f(x)=0$ otherwise), and comparison functions \cite{boyle_function_2015,kumar_compact_2024}.
Other function families such as decision trees \cite{boyle_function_2016} have FSS schemes, but they have fewer applications than DPF and DCF.

All functions within a function family must share the same domain and output space.
Moreover, we use the notation $\mathcal{K}$ to refer to the FSS key space.

\begin{definition}
  A $p$-party FSS scheme (for a family $\mathcal{F}$) has 3 algorithms:
  \begin{itemize}
    \item $\Gen: \mathbb{N} \times P_\mathcal{F} \rightarrow \mathcal{K}^p$ takes as input a security parameter $1^\lambda \in\mathbb{N}$ and a function description $\hat{f}\in P_\mathcal{F}$, and outputs $p$ keys $k_1, \dots, k_p\in \mathcal{K}$.
    \item $\Eval: \mathcal{K} \times \mathcal{X} \rightarrow \mathbb{G}$ takes as input $k_i$ and a point $x\in\mathcal{X}$, outputs a share of $f(x)$ that we denote as $\share{f(x)}_i$.
    \item $\Decode: \mathbb{G}^p \rightarrow  \mathcal{Y}$ takes as input $p$ shares $\{\share{f(x)}_1, \dots,\share{f(x)}_p \}$ and outputs the secret $f(x)$.
  \end{itemize} 
\end{definition}

Recent FSS works \cite{bunn_cnf-fss_2022,zyskind_high-throughput_2024} generalized this definition to support threshold secret-sharing.
Our work (like most existing works \cite{boyle_function_2015,boyle_function_2016,bunn_cnf-fss_2022,corrigan-gibbs_riposte_2015}) does not require such a generalization, so we stick to the classic definitions from \cite{boyle_function_2015}.

\begin{definition}[Correctness \cite{boyle_function_2015}]
  For any function $f\in\mathcal{F}$, for any point $x\in\mathcal{X}$, we have $k_1, \dots, k_p \leftarrow \Gen(1^\lambda, \hat{f})$ and $$
  \mathbb{P}\left[\text{Decode}(\Eval(k_1, x), \dots, \Eval(k_p, x)) = f(x)\right] = 1$$
\end{definition}

\begin{definition}[Privacy \cite{boyle_function_2022}]
  \label{def:privacy}
  Let $\Leak : \{0, 1\}^* \rightarrow \{0, 1\}^*$ be a function specifying the allowable leakage. If omitted, it is understood to be $\Leak(\hat{f}) = (\mathcal{X}, \mathcal{Y})$.

  We call a $p$-party FSS scheme private if, for every set of corrupted parties $S \subseteq \{1\dots p\}$ of size $m$, there exists a PPT algorithm $\Sim$ (simulator), such that for every sequence of function descriptions from $P_\mathcal{F}$ ($\hat{f}_1, \hat{f}_2, \dots$) of size polynomial in $\lambda$, the outputs of the following experiments $\Real$ and $\Ideal$ are computationally indistinguishable:

  \begin{itemize}
    \item $\Real(1^\lambda): (k_1, \dots, k_p) \leftarrow \Gen(1^\lambda, \hat{f}_\lambda); \Output~ (k_i)_{i\in S}$
    \item $\Ideal(1^\lambda): \Output~ \Sim(1^\lambda, \Leak(\hat{f}_\lambda))$
  \end{itemize}
\end{definition}

\section{DDH-based Distributed Point Function}
\label{sec:dpf}
This section introduces a DDH-based approach to build a scheme with $O(\sqrt[3]{N})$ key size upon the information-theoretic scheme of \cite{bunn_cnf-fss_2022} which has an $O(\sqrt[2]{N})$ key size.
We call ``sub-DPF'' or ``sub-scheme'' the DPF scheme of \cite{bunn_cnf-fss_2022} on which we apply our DDH-based optimization.

As this section focuses on point functions, we use the notation $(\alpha, \beta)$ to refer to the parameters of the secret function $f$ (i.e., $f(x) = \beta$ if $x=\alpha$, $0$ otherwise).

\subsection{Scheme}
We represent the domain as a grid of dimensions: $(\sqrt[3]{N})^2\times\sqrt[3]{N}$ (see Figure \ref{fig:ec-ddh-key-dpf})
The non-zero value is in cell $(\gamma_*, \delta_*)$.

\begin{figure}
  \centering
  \includegraphics[width=.6\linewidth]{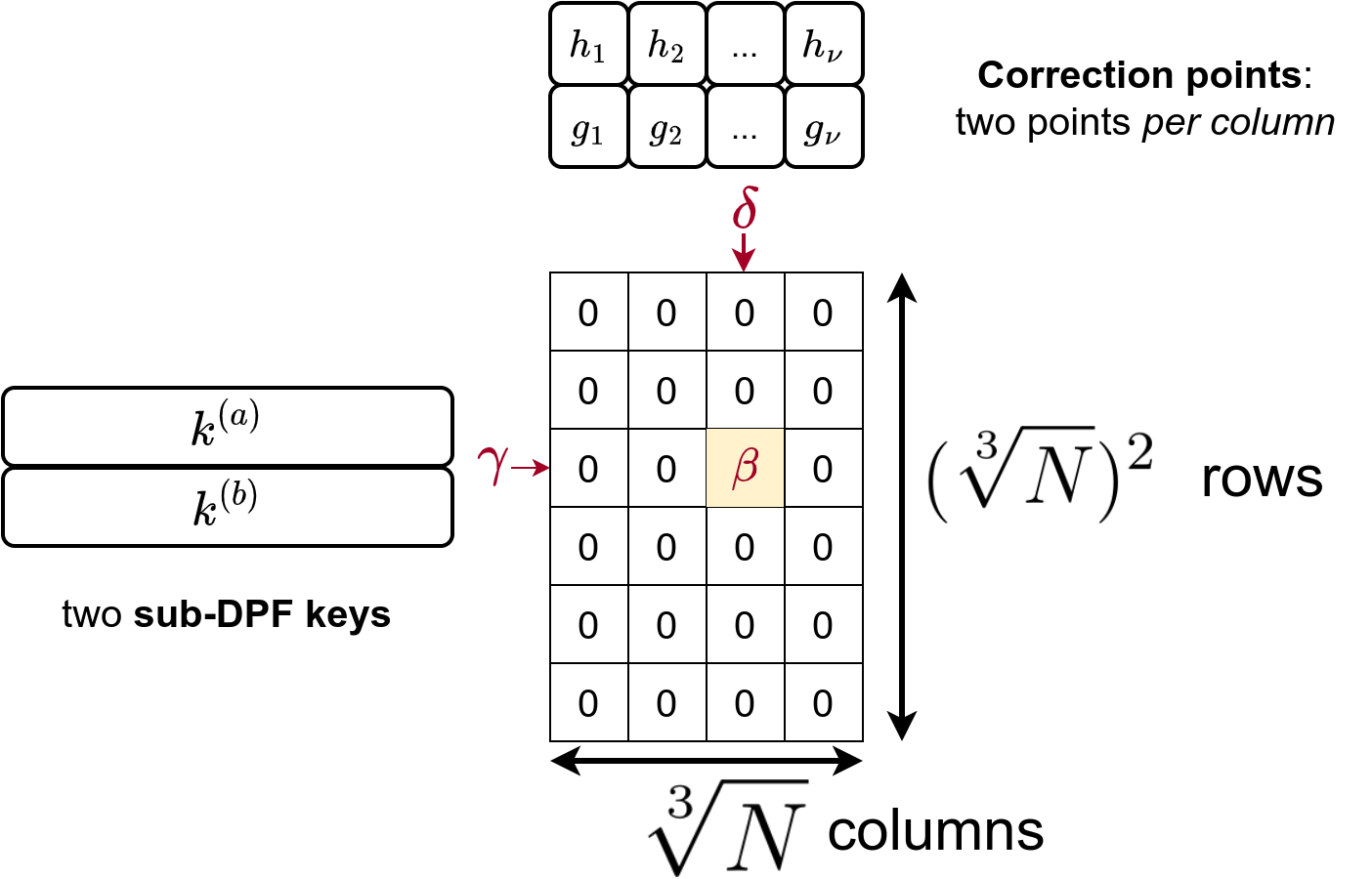}
  \caption{High-level structure of our DDH-based DPF}
  \label{fig:ec-ddh-key-dpf}
\end{figure}

$\Gen$:
Let $\mathbb{G}$ be a cyclic group of prime order $q_0$, with $g$ a generator. We assume that the DDH assumption holds in $\mathbb{G}$.
Let $g_\beta=\Encode_{\DDH}(\beta)$ be the encoding of $\beta$ in $\mathbb{G}$ (possible encodings are discussed in Section \ref{sec:secret-encoding}).
Our $\Gen$ algorithm starts by sampling one random number $r \xleftarrow{R} \mathbb{F}_{q_0}$ and compute $r_\text{inv}$ its multiplicative inverse.
This multiplicative inverse exists because $q_0$ is prime.

We then use the sub-DPF to share two secret point functions $f_a, f_b$ such that $f_a(x) = r$ and $f_b(x) = 1$, if $x=\gamma_*$, otherwise $0$.
This step outputs two sets of DPF keys $\{k_1^{(a)},\dots, k_p^{(a)}\}$ and $\{k_1^{(b)},\dots, k_p^{(b)}\}$, with each key of size $O(\sqrt[3]{N})$ thanks to the sub-DPF.
For each column $\delta$, we sample a random generator $g_{\delta} \xleftarrow{R} \mathbb{G}$.
For all $\delta\neq\delta_*$,  we set the ``correction points'' $h_\delta = g_\delta^{r_\text{inv}}$.
We set the correction points for the column $\delta_*$: $h_{\delta_*} \leftarrow g_{\delta_*}^{-r_\text{inv}} \cdot g_\beta^{r_\text{inv}}$.
Each party $i$ receives a key $k_i$ containing two DPF ``sub-keys'' $k_i^{(a)}, k_i^{(b)}$ and all corrections points.

$\Eval$:
represents $x$ as $(\gamma', \delta')$ and evaluates: $\share{s_a}_i \leftarrow \Eval_{\DPF}(k_i^{(a)}, \gamma')$ and $\share{s_b}_i \leftarrow \Eval_{\DPF}(k_i^{(b)})$.
Each party obtains $\share{f(x)}_i = h_{\delta'}^{\share{s_a}_i} \cdot g_{\delta'}^{\share{s_b}_i}$.

Algorithm \ref{alg:dpf-ddh} details our DDH-based scheme.

\begin{algorithm}[t]
  \caption{DDH-based DPF scheme}\label{alg:dpf-ddh}
  \begin{algorithmic}[1]
    \State{Let $\mathbb{G}$ be a cyclic group of prime order $q_0$ and with $g$ a generator. We assume that the DDH assumption holds in $\mathbb{G}$.}
    \State{Let $(\Gen^{(*)}_{\DPF}, \Eval^{(*)}_{\DPF}, \Decode_{+})$ be the information-theoretic DPF proposed by Bunn et al. \cite{bunn_cnf-fss_2022} (with output in $\mathbb{F}_{q_0}$).}
    \State{Let $\Encode_{\DDH}: \mathcal{Y} \rightarrow \mathbb{G}$ and $\Decode_{\DDH}: \mathbb{G} \rightarrow \mathcal{Y} $ be two functions such that $\Encode_{\DDH}(0)=g^0$ and $\Decode_{\DDH}(\Encode_{\DDH}(y)) = y$ for any $y\in\mathcal{Y}$.}
    \Statex{}
    \Function{$\Gen_{\DPF}$}{$\alpha, \beta, p, m$}
    \State{Let $g_\beta\leftarrow \Encode_{\DDH}(\beta)$.}
    \State{Let $(\gamma_*, \delta_*)$ be the position of $\alpha$ in a $(\nu)^2 \times \nu$ grid, with $\nu= \lceil\sqrt[3]{N}\rceil$}

    \State{Sample one (non-zero) $r \xleftarrow{R} \mathbb{F}_{q_0}$ and set $r_\text{inv}\leftarrow r^{-1}$.}
    \State{Generate two sets of DPF keys:}
    \Statex{\hspace{3em}$\{k_1^{(a)},\dots, k_p^{(a)}\}\leftarrow \Gen^{(*)}_{\DPF}(\gamma_*, r, p, m)$} 
    \Statex{\hspace{3em}$\{k_1^{(b)},\dots, k_p^{(b)}\}\leftarrow \Gen^{(*)}_{\DPF}(\gamma_*, 1, p, m)$.}

    \For{$\delta\in\{1\dots \nu\}, \delta\neq\delta_*$}
    \State{Sample a point $g_\delta$ from $\mathbb{G}$ and set $h_\delta \leftarrow g_\delta^{-r_\text{inv}}$.}
    \EndFor
    \State{Sample a $g_{\delta_*}$ from $\mathbb{G}$ and set $h_{\delta_*} \leftarrow g_{\delta_*}^{-r_\text{inv}} \cdot g_\beta^{r_\text{inv}}$.}
    \State{Set $k_i = (k_i^{(a)}||k_i^{(b)}||g_1||h_1||\dots||g_{\nu}||h_{\nu}), \forall i \in\{1\dots p\}$}
    \State\Return{$(k_1, \dots, k_p)$.}
    \EndFunction
    \Statex{}

    \Function{$\Eval_{\DPF}$}{$k_i, x$}
    \State{Let $(\gamma', \delta')$ be the position of $x$ in a $(\sqrt[3]{N})^k \times \sqrt[3]{N}$ grid.}
    \State{Parse $k_i$ as $k_i= (k_i^{(a)}||k_i^{(b)}||g_1||h_1||\dots||g_{\nu}||h_{\nu})$.}
    \State{Let $\share{s_a}_i \leftarrow \Eval^{(*)}_{\DPF}(k_i^{(a)}, \gamma'), \share{s_b}_i \leftarrow \Eval^{(*)}_{\DPF}(k_i^{(b)},\gamma')$.}
    \State{Let $\share{f(x)}_i \leftarrow h_{\delta'}^{\share{s_a}_i} \cdot g_{\delta'}^{\share{s_b}_i}$.}
    \State\Return{$\share{f(x)}_i$}
    \Comment{$\share{f(x)}_i\in \mathbb{G}$.}
    \EndFunction

    \Statex
    \Function{$\Decode_{\DPF}$}{$\share{f(x)}_1, \dots, \share{f(x)}_p$}
    \Return{$\Decode_{\DDH}(\prod_{i=1}^p \share{f(x)}_i)$}
    \EndFunction
  \end{algorithmic}
\end{algorithm}

\subsection{Asymptotic key size}
Figure \ref{fig:ec-ddh-key-dpf} provides a high-level representation of our DPF keys, with an $O(\sqrt[3]{N} \cdot \binom{p-1}{m} \cdot (\lambda + \log q))$ key size.
Let us break down this asymptotic cost.

First, we have two sub-DPF keys.
The previous DPFs are defined over the domain $\{1\dots (\sqrt[3]{N})^2\}$ to ``cover all the rows''.
We have $(\sqrt[3]{N})^2$ rows and use \cite{bunn_cnf-fss_2022} as sub-DPF (with an $O(\sqrt{M}\binom{p-1}{m}\log q_0)$ key size for a domain of size $M$).
Each sub-DPF key is then of size $O(\sqrt[3]{N} \cdot \binom{p-1}{m} \cdot \log q_0)$ ($M=(\sqrt[3]{N})^2$ in our case).

Second, for each column $\delta$, we have $g_\delta, h_\delta\in \mathbb{G}$.
For convenience, assume $\mathbb{G} = \mathbb{F}^\times_q$.
As we have $\sqrt[3]{N}$ columns, the total size of these elements is $O(\sqrt[3]{N} \cdot \log q)$.

In our scheme, $\lambda$ is implicitly present because DDH is assumed to be hard in cyclic group $\mathbb{G}$.
As $\mathbb{G}$ is of order $q_0$, we have $\lambda = O(\log q_0)$.

\paragraph{Supporting other sub-DPFs}
We use a grid of size $(\sqrt[3]{N})^2\times\sqrt[3]{N}$ because the sub-DPF we use has an $O(\sqrt{M})$ key size (for a domain size of $M$).
However, using our DDH-based approach, we can replace \cite{bunn_cnf-fss_2022} with any scheme that outputs additive shares in $\mathbb{F}_{q_0}$. 
Specifically, for any sub-scheme with key size $O(\sqrt[k]{N})$, our method yields a DDH-based variant with key size $O(\sqrt[k+1]{N})$.

In particular, our approach can be applied to the dishonest-majority scheme proposed in \cite{goel_multiparty_2025}.
The resulting scheme (combining \cite{goel_multiparty_2025} with our DDH-based approach) is secure against a dishonest majority and has key size $O(\sqrt[3]{N} \cdot p^3 \lambda^3 \cdot (\log q + \lambda))$.
As shown in Section~\ref{sec:benchmark}, the scheme of \cite{goel_multiparty_2025} already has larger keys than sharing the function’s truth table.
Therefore, extending our scheme using \cite{goel_multiparty_2025} would result in an impractical solution, because \cite{goel_multiparty_2025} is impractical.

Given this, we do not further develop this extension in our work.
Instead, we focus our discussions solely on our honest-majority scheme, which, in contrast, provides substantial key size reductions for practical problem sizes.

Since we can reduce key sizes from $O(\sqrt[k]{N})$ to $O(\sqrt[k+1]{N})$, our result raises a key question: can we build a multi-party scheme with logarithmic key size via a recursive application of our scheme?
Our DDH-based scheme outputs multiplicative shares (in a DDH group), so it cannot be used recursively.

\subsection{Security}
\label{subsec:sec-ddh-dpf}
In our scheme, the DDH assumption prevents an adversary from distinguishing the correction points, and then prevents them from recovering the column $\delta_*$ containing the non-zero from the key.
To prove this intuition, we propose to consider the following theorem:

\begin{theorem}
  \label{th:security}
  Let $\lambda \in \mathbb{N}$, $N, p \in \mathbb{N}$, then $(\Gen_{\DPF}, \Eval_{\DPF}, \Decode_{\DPF})$ as described in Algorithm \ref{alg:dpf-ddh} is an FSS scheme for the family of point functions.

  Assuming that the DDH assumption holds in $\mathbb{G}$ and that the information-theoretic scheme $(\Gen^{(*)}_{\DPF}, \Eval^{(*)}_{\DPF}, \Decode_{+})$ from \cite{bunn_cnf-fss_2022} is a correct and private DPF scheme, then this scheme is correct and private against at most $m$ semi-honest parties with $m < p/2$.
\end{theorem}
\begin{proof}
  Given in Appendix \ref{app:proof-th:security}.
\end{proof}

\subsection{Comparison with existing works}

\begin{table}
  \newcolumntype{P}[1]{>{\centering\arraybackslash}p{#1}}
  \renewcommand*{\arraystretch}{1.8}
  \begin{center}
    \begin{tabular}{|P{2.7cm}|c|c|P{3cm}|P{3.8cm}|}
      \hline
      Scheme                            & Year & Majority  & Assumptions                     & Key size                                                                           \\\hline\hline
      \cite{boyle_function_2015}        & 2015 & Dishonest & PRG                             & $O(\sqrt{N}\cdot q^{\frac{p-1}{2}}\cdot(\log q + \lambda))$                        \\\hline
      \cite{corrigan-gibbs_riposte_2015,kumar_compact_2024} & 2015 & Dishonest & DDH                             & $O(\sqrt{N}\cdot(\lambda + \log q))$                                               \\\hline
      PRG-based \cite{bunn_cnf-fss_2022} & 2022 & Honest    & PRG                             & $O(\sqrt[4]{N} \cdot \sqrt{q^{p^m}} \cdot \binom{p-1}{m} \cdot(\lambda + \log q))$ \\\hline
      Info.-Th. \cite{bunn_cnf-fss_2022} & 2022 & Honest    & None                            & $O(\sqrt{N}\cdot \binom{p-1}{m} \cdot \log q)$                                     \\\hline
      \cite{goel_multiparty_2025} & 2025 & Dishonest    & PRG                             & $O(\sqrt{N}\cdot {p^3\lambda^3}\cdot(\log q + \lambda))$                                             \\\hline
      \cite{krips_multi-party_2025} & 2025 & Dishonest    & DDH \textbf{+ 2 new hardness assumpt.}                             & $O(\log{N}\cdot (\log q + \lambda))$                                             \\\hline\hline
      \textbf{\color{blue}Our scheme} & 2025    & Honest    & DDH                             & $O(\sqrt[3]{N} \cdot \binom{p-1}{m} \cdot (\lambda + \log q))$                     \\\hline
    \end{tabular}
  \end{center}
  \caption{Comparison of the multi-party DPF schemes}
  \label{tab:dpf}
\end{table}

Table \ref{tab:dpf} compares existing schemes to ours.
We achieve two goals: (1) avoiding exponential factors present in \cite{boyle_function_2015,bunn_cnf-fss_2022}, and (2) improving the dependency on the function domain size $N$ (compared to other practical schemes \cite{bunn_cnf-fss_2022,corrigan-gibbs_riposte_2015,goel_multiparty_2025}).
We offer an intermediary solution between an $O(\sqrt[4]{N})$ key size with exponential factors \cite{bunn_cnf-fss_2022} and $O(\sqrt{N})$ key sizes without exponential factors \cite{bunn_cnf-fss_2022,corrigan-gibbs_riposte_2015,goel_multiparty_2025}.

The scheme of \cite{krips_multi-party_2025} achieves the best known asymptotic key size, namely $O(\log N)$.
However, this improvement comes at the cost of relying on two new hardness assumptions whose concrete security is not yet well understood.
The work of \cite{krips_multi-party_2025} represents an exciting theoretical advancement, but the lack of concrete security hinders its implementation \emph{for now}.
Further research is needed to determine parameter sets that achieve standard security levels (e.g., 128-bit or 256-bit).
In contrast, our scheme has $O(\sqrt[3]{N})$ key sizes, but is based solely on the DDH assumption; a \emph{standard and extensively studied cryptographic assumption}.

\subsection{Pseudo-random secret sharing}
Pseudo-random secret sharing (PRSS) \cite{cramer_share_2005} is a protocol reducing the communication cost of secret sharing protocols thanks to pseudo-random seed expansion.
As most multi-party schemes \cite{boyle_function_2015,bunn_cnf-fss_2022,corrigan-gibbs_riposte_2015} (including ours) output DPF keys containing secret shares, we can reduce their size thanks to PRSS.

In particular, we rely on the simplest form of PRSS applied to additive secret shares.
Additive secret sharing on a secret $s$ normally outputs $p$ shares $\{\share{s}_1 \dots \share{s}_p\}$, such that $\sum_i \share{s}_i = s$.
With PRSS, the secret holder generates $p-1$ random seeds $r_1, \dots, r_{p-1}$ and a share $\share{s}_p$, such that $\share{s}_p +\sum_{i=1}^{p-1} G(r_i) = s $ (with $G$ a PRG).
One shareholder receives $\share{s}_p$ and the other $p-1$ parties receive a random seed $r_i$.

The random seeds can be reused to share multiple values.
To share a vector of $n$ values, the secret holder sends $p-1$ seeds and a single vector of $n$ shares (vs. $p$ vector of $n$ shares in classic secret sharing).
If the MPC protocol allows the secret holder to be a shareholder, they can keep the vector of shares and send only random seeds to the other shareholders.

This technique amortizes the communication costs (i.e., key size) when the protocols require sharing large vectors.
Since all practical schemes \cite{bunn_cnf-fss_2022,corrigan-gibbs_riposte_2015} are compatible with PRSS, Section \ref{sec:benchmark} compares the key sizes using PRSS.
\section{Key size benchmark}
\label{sec:benchmark}
We compare the exact key size of our DPF scheme to those of existing schemes.
Our experiments report the total key size instead of the individual key size.
This metric takes into account the key size amortization provided by PRSS.
For this benchmark, we implement our scheme using elliptic curves $E(\mathbb{F}_q)$, as they provide prime-order cyclic groups in which the DDH is known to be hard.
More specifically, we use the curve P-256.

\paragraph{Exponential factors in existing works}
Figure \ref{fig:exponential_factor} illustrates the exponential factor of the PRG-based multi-party DPF \cite{boyle_function_2015,bunn_cnf-fss_2022}.
Figure \ref{subfig:prime-moduli} represents their key size in function of the output bit length (i.e., $\log_2 q$ for an output in $\mathbb{F}_q$).
While this problem was barely discussed in existing works \cite{boyle_function_2015,bunn_cnf-fss_2022}, Figure \ref{fig:exponential_factor} shows that it makes these techniques purely unusable.
For only 12-bit moduli, the DPF of \cite{boyle_function_2015} has a key size of $10^{18}$ bits and the PRG-based DPF of \cite{bunn_cnf-fss_2022} has a key size of $10^{38}$.

\begin{figure}[t]
  \centering
  \begin{subfigure}{0.4\linewidth}
    \includegraphics[width=\linewidth]{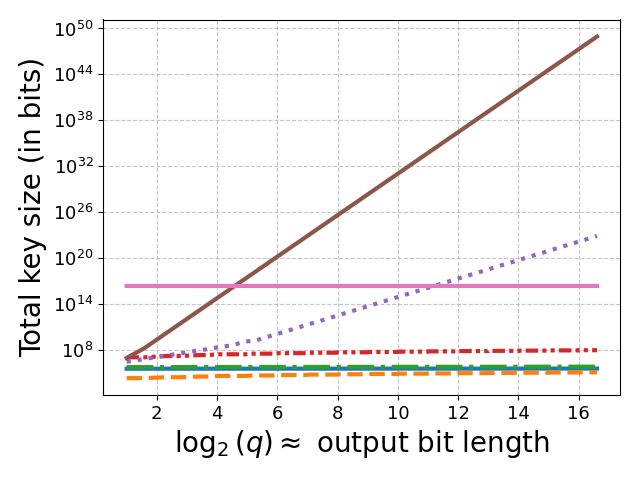}
    \caption{Prime moduli}
    \label{subfig:prime-moduli}
  \end{subfigure}
  \begin{subfigure}{0.4\linewidth}
    \includegraphics[width=\linewidth]{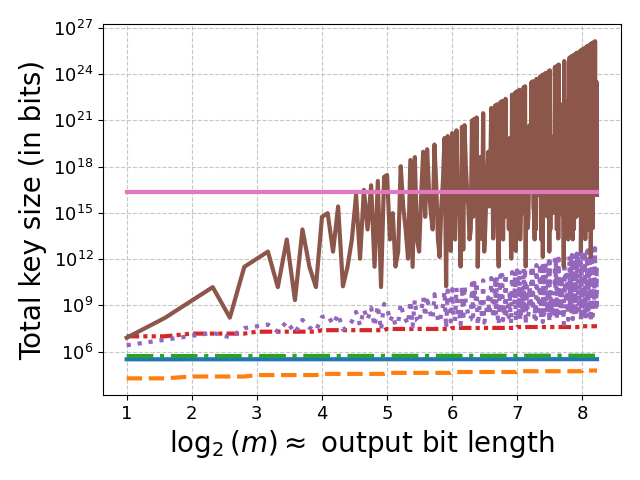}
    \caption{Arbitrary moduli}
    \label{subfig:arbitrary-moduli}
  \end{subfigure}
  \begin{subfigure}{0.4\linewidth}
    \includegraphics[width=\linewidth]{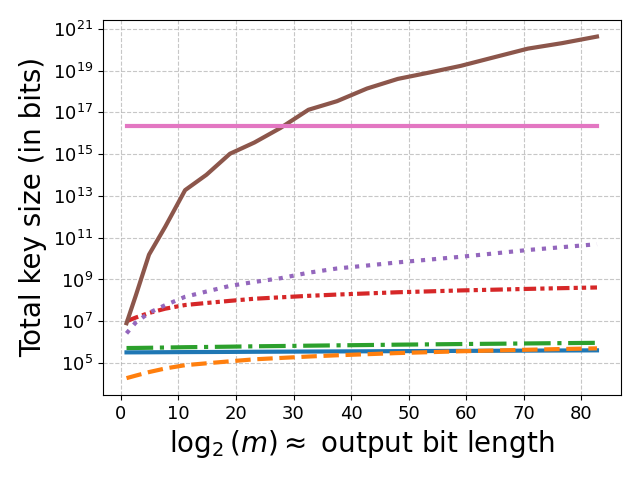}
    \caption{Primorial moduli}
    \label{subfig:primorial-moduli}
  \end{subfigure}
  \begin{subfigure}{0.4\linewidth}
    \raisebox{0.6\height}{\includegraphics[width=\linewidth]{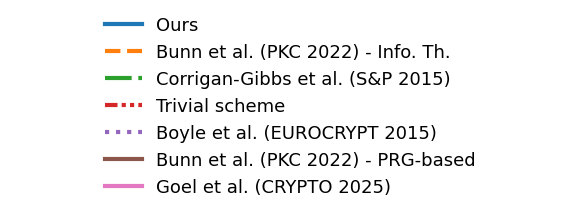}}
    \caption{Legend}
  \end{subfigure}
  \caption{DPF key sizes for varying moduli. Parameters: $p=5$ parties, $N=10^6$.}
  \label{fig:exponential_factor}
\end{figure}

Figure \ref{fig:exponential_factor} includes a curve ``Trivial scheme'', corresponding to the most trivial DPF implementation: sharing the truth table of the function.
This trivial baseline is much better than these PRG-based solutions for prime fields other than $\mathbb{F}_2$.
While other works \cite{boyle_function_2022,kumar_compact_2024} warned about this potential scalability issue, our work provides the first evidence of their impracticality in $\mathbb{F}_q$.

In comparison, we can barely see on Figure \ref{subfig:prime-moduli} the key size difference between the rest of the DPF schemes \cite[Ours]{bunn_cnf-fss_2022,corrigan-gibbs_riposte_2015}.
Note that Figure \ref{fig:exponential_factor} includes curves for the DDH-based schemes \cite[Ours]{corrigan-gibbs_riposte_2015}.
Such schemes are implemented with a fixed elliptic curve P-256, so the output is of fixed size (256 bits).
Thus, their key sizes are constant on this figure.

Recently, Boyle \cite{boyle_function_2022} proposed an optimization based on the Chinese Remainder Theorem (CRT) to optimize the key sizes of PRG-based schemes \cite{boyle_function_2015,bunn_cnf-fss_2022} for composite moduli.
Thanks to this CRT trick, we can replace the exponential factor $q^p$ with a sum of smaller factors $\sum_i q_i^p$ (with $q_i$ the prime factors of the composite modulus).

Figure \ref{subfig:arbitrary-moduli} represents the key sizes for varying composite moduli, but the figure is hard to read, because this trick is sensitive to the prime decomposition.
Two numbers $m$ and $m+1$ can have completely different prime factors.
To provide more readable results, Figure \ref{subfig:primorial-moduli} represents the key sizes for varying priomorial moduli.
A primorial is a composite number whose prime factors are the $n$ smallest primes.
Such numbers are the best case scenario for the CRT trick because they provide the smallest factors possible.
On Figure \ref{subfig:primorial-moduli}, we observe that the trick significantly improves the key size of the PRG-based schemes, but they remain way above the trivial scheme for any modulus above 210.

Finally, Figure~\ref{fig:exponential_factor} provides a key insight into \cite{goel_multiparty_2025}.
Note that their key size may seem constant in $q$ on Figure \ref{fig:exponential_factor}, but it is not; their key size is simply dominated by factors independent of $q$.
While their work replaces the exponential factor present in \cite{boyle_function_2015} with a polynomial factor of $p^3\lambda^3$, we demonstrate that, for realistic problem sizes, this polynomial term actually exceeds $q^p$.
As a result, the asymptotic improvement offered by \cite{goel_multiparty_2025} does not translate into practical efficiency: the key size remains significantly larger than that of the trivial scheme.
In other words, although \cite{goel_multiparty_2025} presents a notable theoretical advance by eliminating the exponential term of \cite{boyle_function_2015}, the resulting scheme remains impractical; performing worse than the trivial DPF scheme.

Since PRG-based schemes \cite{boyle_function_2015,bunn_cnf-fss_2022,goel_multiparty_2025} are impractical for moduli larger than 210 (even with the CRT optimization of \cite{boyle_function_2022}), we exclude them from our other experiments.
This allows us to observe precisely the key size improvements achieved by our scheme when compared to practical DPF schemes \cite{bunn_cnf-fss_2022,corrigan-gibbs_riposte_2015}.

\paragraph{Varying function domain}
Figure \ref{fig:function_domain_size} compares the key sizes in function of the domain sizes; varying from $10^2$ to $10^{10}$.
Zyskind et al. \cite{zyskind_high-throughput_2024} used the same problem sizes to benchmark three-party DPF schemes.

\begin{figure}[t]
  \centering
  \begin{subfigure}{0.45\linewidth}
    \includegraphics[width=\linewidth]{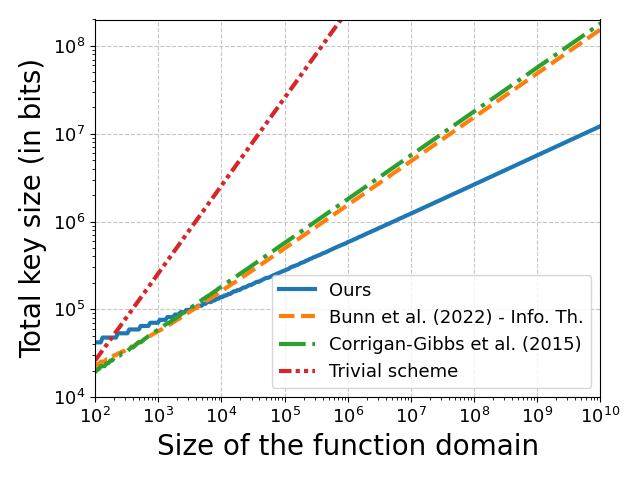}
    \caption{Varying domain size ($p=5$)}
    \label{fig:function_domain_size}
  \end{subfigure}
  \begin{subfigure}{0.45\linewidth}
    \includegraphics[width=\linewidth]{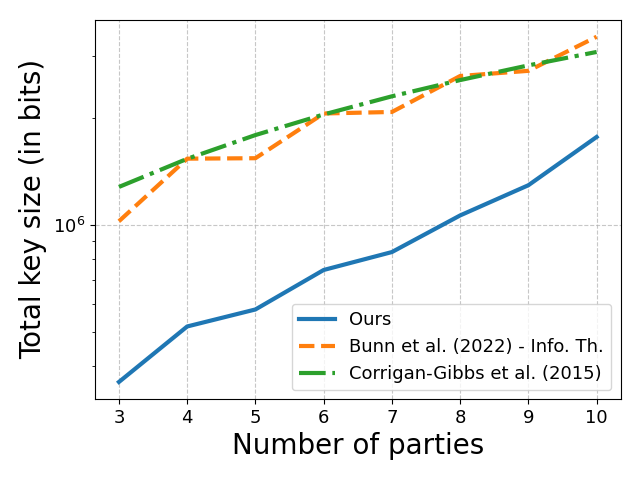}
    \caption{ Varying nb. of parties ($N=10^6$)}
    \label{fig:nb_parties}
  \end{subfigure}
  \caption{Key sizes for varying domain sizes and number of parties.}
\end{figure}

Our key size is the smallest for any domain size above $10^3$.
When $N=10^6$, our key size is 3 times smaller.
When $N=10^9$, our key size is 10 times smaller.

Our key size is higher than those of \cite{bunn_cnf-fss_2022,corrigan-gibbs_riposte_2015} only when $N$ is below $10^3$.
To put this into perspective, Figure \ref{fig:function_domain_size} also shows that the trivial DPF scheme becomes the most efficient for smaller domain sizes (below $10^2$).
Thus, the advantage provided \cite{bunn_cnf-fss_2022,corrigan-gibbs_riposte_2015} only holds for particularly small function domains and the trivial DPF could even be preferable in such cases.

\paragraph{Varying number of parties}
Figure \ref{fig:nb_parties} compares the key sizes in function of the number of parties; varying from 3 to 10.
We observe that Corrigan-Gibbs et al. \cite{corrigan-gibbs_riposte_2015} has a better scaling than our scheme, but our key size remains smaller.

The growth of our key size is not as smooth as the growth of the scheme by \cite{corrigan-gibbs_riposte_2015}.
For example, the growth between 3 and 4 parties is steeper than between 4 and 5.
This phenomenon is caused by the honest-majority assumption: (contrary to \cite{corrigan-gibbs_riposte_2015}) our total key size is correlated to the number of honest parties.
Under the honest-majority assumption, 4-party and 5-party setups have exactly the same number of honest parties (i.e., 3).
The same phenomenon is present for \cite{bunn_cnf-fss_2022}.




\section{DDH-based Distributed Comparison Function}
\label{sec:dcf}
This section adapts our DPF scheme to comparison functions.
However, the literature has described fewer DCF schemes than DPF.
Indeed, only \cite{boyle_function_2015,kumar_compact_2024} have described multiparty DCFs, which provides weak baselines compared to DPF.
Thus, Appendix \ref{subsec:adapt-bunn-et-al} adapts the DPF schemes of \cite{bunn_cnf-fss_2022} to DCF in order to compare our DCF to challenging baselines. 

\subsection{Scheme}
We build our DDH-based DCF upon the information-theoretic DPF of \cite{bunn_cnf-fss_2022} and the DCF adapted from \cite{bunn_cnf-fss_2022} (see Appendix \ref{subsec:adapt-bunn-et-al}).
The main change compared to our DPF scheme is that we require both a sub-DCF and a sub-DPF.
Algorithm \ref{alg:dcf-ddh} details our scheme (with the changes compared to our DPF in {\color{blue}blue}).



To understand their respective roles, let us represent the function domain as a grid.
On the one hand, the DCF covers all rows containing only non-zeros values (i.e., $\beta$).
On the other hand, the DPF enables to select the row $\gamma_*$ containing a segment of $\beta$ and a segment of $0$.
To complete our extension, we update the definition of the ``correction points'', so that the $\gamma_*$ row contains $\beta$ until the column $\delta_*$ and 0 after.

\begin{algorithm}[t]
  \caption{DDH-based DCF scheme (changes compared to our DPF in {\color{blue}blue})}\label{alg:dcf-ddh}
  \begin{algorithmic}[1]
    \State{Let $\mathbb{G}$ be a cyclic group of prime order $q_0$ and with $g$ a generator.}
    \State{Let $(\Gen^{(*)}_{\DPF}, \Eval^{(*)}_{\DPF}, \Decode_{+})$ be the information-theoretic DPF of Bunn et al. \cite{bunn_cnf-fss_2022} {\color{blue} and $(\Gen^{(*)}_{\DCF}, \Eval^{(*)}_{\DCF}, \Decode_+)$ be the DCF adapted from it}.}
    \State{Let $\Encode_{\DDH}: \mathcal{Y} \rightarrow \mathbb{G}$ and $\Decode_{\DDH}: \mathbb{G} \rightarrow \mathcal{Y} $ be two functions such that $\Encode_{\DDH}(0)=g^0$ and $\Decode_{\DDH}(\Encode_{\DDH}(y)) = y$ for any $y\in\mathcal{Y}$.}
    \Statex{}
    \Function{$\Gen_{\DPF}$}{$\alpha, \beta, p, m$}
    \State{Let $g_\beta\leftarrow \Encode_{\DDH}(\beta)$.}
    \State{Let $(\gamma_*, \delta_*)$ be the position of $\alpha$ in a $(\nu)^k \times \nu$ grid, with $\nu= \lceil\sqrt[3]{N}\rceil$.}
    \State{Sample {\color{blue}two} (non-zero) $r, s \xleftarrow{R} \mathbb{F}_{q_0}$ and set $r_\text{inv}\leftarrow r^{-1}$ {\color{blue}and $s_\text{inv}\leftarrow {s}^{-1}$}.}
    \State{Generate two sets of DPF keys {\color{blue}and one of DCF keys}:}
    \Statex{\hspace{3em}$\{k_1^{(a)},\dots, k_p^{(a)}\}\leftarrow \Gen^{(*)}_{\DPF}(\gamma_*, r, p, m)$}
    \Statex{\hspace{3em}$\{k_1^{(b)},\dots, k_p^{(b)}\}\leftarrow \Gen^{(*)}_{\DPF}(\gamma_*, 1, p, m)$}
    \Statex{\hspace{3em}\color{blue}$\{k_1^{(c)},\dots, k_p^{(c)}\}\leftarrow \Gen^{(*)}_{\DCF}(\gamma_*-1, r_c, p, m)$.}

    \State{\textbf{for} {\color{blue}$\delta > \delta_*$}} \textbf{do} {Sample a point $g_\delta$ from $\mathbb{G}$ and set $h_\delta \leftarrow g_\delta^{-r_\text{inv}}$.}
    \State{\textbf{for} {\color{blue}$\delta\leq\delta_*$}} \textbf{do} {Sample a $g_{\delta}$ from $\mathbb{G}$ and set $h_{\delta} \leftarrow g_{\delta}^{-r_\text{inv}} \cdot g_\beta^{r_\text{inv}}$.}

    \State{\color{blue}Let $u \leftarrow g_\beta^{s_\text{inv}}$.}
    \State{Set $k_i = (k_i^{(a)}||k_i^{(b)}||g_1||h_1||\dots||g_{\nu}||h_{\nu}{\color{blue}||u}), \forall i \in\{1\dots p\}$}
    \State\Return{$(k_1, \dots, k_p)$.}
    \EndFunction
    \Statex{}

    \Function{$\Eval_{\DPF}$}{$k_i, x$}
    \State{Let $(\gamma', \delta')$ be the position of $x$ in a $(\sqrt[3]{N})^2 \times \sqrt[3]{N}$ grid.}
    \State{Parse $k_i=(k_i^{(a)}||k_i^{(b)}||g_1||h_1||\dots||g_{\nu}||h_{\nu}{\color{blue}||u})$.}
    \State{Let $\share{s_a}_i \leftarrow \Eval^{(*)}_{\DPF}(k_i^{(a)}, \gamma'), \share{s_b}_i \leftarrow \Eval^{(*)}_{\DPF}(k_i^{(b)},\gamma')$}
    \State{\color{blue} Let $\share{s_c}_i \leftarrow \Eval^{(*)}_{\DCF}(k_i^{(c)},\gamma')$}
    \State{Let $\share{f(x)}_i \leftarrow h_{\delta'}^{\share{s_a}_i} \cdot g_{\delta'}^{\share{s_b}_i}{\color{blue}\cdot u^{\share{s_c}_i}}$.}
    \State\Return{$\share{f(x)}_i$}
    \Comment{$\share{f(x)}_i\in \mathbb{G}$.}
    \EndFunction

    \Statex
    \Function{$\Decode_{\DPF}$}{$\share{f(x)}_1, \dots, \share{f(x)}_p$}
    \Return{$\Decode_{\DDH}(\prod_{i=1}^p \share{f(x)}_i)$}
    \EndFunction
  \end{algorithmic}
\end{algorithm}

This scheme is a simple adaptation of our DPF; adding a (secure) sub-DCF and updating the correction points.
Thus, its security proof relies on the same arguments as our DPF scheme (see Appendix \ref{app:proof-th:security}).

\subsection{Comparison}
Table \ref{tab:dcf} compares our DCF schemes to the existing schemes.
The DCF schemes adapted from \cite{bunn_cnf-fss_2022} have respectively $O(\sqrt[4]{N} \cdot \sqrt{q^{p^m}} \cdot \binom{p-1}{m} \cdot(\lambda + \log q))$ key size for the PRG-based and $O(\sqrt{N}\cdot \binom{p-1}{m} \cdot \log q)$ key size for the information-theoretic.
Finally, our DDH-based scheme has a key size of $O(\sqrt[3]{N}\cdot \binom{p}{m+1}\cdot \log q)$.

\begin{table}[t]
  \newcolumntype{P}[1]{>{\centering\arraybackslash}p{#1}}
  \renewcommand*{\arraystretch}{1.8}
  \begin{center}
    \begin{tabular}{|P{2.9cm}|c|P{2.6cm}|P{4.2cm}|}
      \hline
      Scheme                                                                 & Majority  & Assumption                     & Key size                                                                           \\\hline\hline
      \cite{boyle_function_2015}                                             & Dishonest & PRG                            & $O(\sqrt{N}\cdot q^{\frac{p-1}{2}} \cdot (\lambda+\log q))$                        \\\hline
      \cite{kumar_compact_2024}                                              & Dishonest & DDH \textbf{+ Trusted decoder} & $O(\sqrt{N} (2^{-\frac{p-1}{2}}\cdot\lambda+ 2^{\frac{p-1}{2}}\cdot\log q))$       \\\hline\hline
      \textbf{\color{blue}Our adaptation} PRG-based \cite{bunn_cnf-fss_2022} & Honest    & PRG                            & $O(\sqrt[4]{N} \cdot \sqrt{q^{p^m}} \cdot \binom{p-1}{m} \cdot(\lambda + \log q))$ \\\hline
      \textbf{\color{blue}Our adaptation} Info.-Th. \cite{bunn_cnf-fss_2022} & Honest    & None                           & $O(\sqrt{N}\cdot \binom{p-1}{m} \cdot \log q)$                                     \\\hline\hline
      \textbf{\color{blue}Our scheme}                                        & Honest    & DDH                            & $O(\sqrt[3]{N} \cdot \binom{p-1}{m} \cdot (\lambda + \log q))$                     \\\hline
    \end{tabular}
    \vspace{1em}

    \caption{Comparison of the multi-party DCF schemes}
    \label{tab:dcf}
  \end{center}
\end{table}

We obtain the same conclusions as in multi-party DPF: our scheme provides an intermediary solution between efficient schemes \cite{kumar_compact_2024,bunn_cnf-fss_2022} with $O(\sqrt{N})$ key size and an impractical PRG-based scheme \cite{bunn_cnf-fss_2022} with $O(\sqrt[4]{N}q^p)$ key size. 
\section{DDH-based secret encodings and their applications}
\label{sec:secret-encoding}
Our schemes (as well as other DDH-based schemes) need to encode/decode the secret value $\beta$ in the cyclic group $\mathbb{G}$.
This encoding and decoding steps are essential as FSS applications usually involve real-valued or integer secrets that must be mapped to $\mathbb{G}$.
Existing papers \cite{corrigan-gibbs_riposte_2015,kumar_compact_2024} have been vague on the exact way to encode and handle secret in DDH-based schemes.

To have a discussion as concrete as possible, this section uses elliptic curves as cyclic group; $\mathbb{G} = E(\mathbb{F}_q)$.
Elliptic curves provide (DDH-hard) cyclic group with prime orders.
They were notably used by \cite{corrigan-gibbs_riposte_2015} as they provide smaller key sizes than other DDH-hard cyclic groups. 

Contrary to other DPF schemes \cite{boyle_function_2015,bunn_cnf-fss_2022}, DDH-based schemes use multiplication to decode the secret shares.
However, using elliptic curves, the multiplication in the cyclic group corresponds to an addition of elliptic curve points.
Even though elliptic curves provide a form of additive decoding, this decoding does not provide the same properties as linear share decoding in $\mathbb{F}_q$.

Thus, this section presents two possible DDH-based secret encoding.
For each encoding, we present its properties and use cases.

\subsection{Share compressibility}
Boyle et al. \cite{boyle_function_2015} identified some desirable share properties, in particular: compressibility.
According to them, a scheme has compressible shares if we can combine the shares ``in a meaningful way'' without communication between the shareholders.
They focused on the aggregation of multiple shares: $\share{f(x_1)} \oplus \share{f(x_2)} = \share{f(x_1) + f(x_2)}$.
Many FSS applications (e.g., PIR \cite{gilboa_distributed_2014} or private histograms \cite{boneh_lightweight_2021}) require a form of compressibility.

Works like \cite{boyle_function_2015,bunn_cnf-fss_2022} satisfy this property thanks to their additive secret-sharing in $\mathbb{F}_q$.
However, ``DDH-based'' encoding can require a non-linear share decoding, that does not necessarily guarantee compressibility.
The next subsections present two DDH-based secret encodings, discuss their relation to compressibility, and list some applications.

\subsection{Point-based secret encoding}
An elliptic curve point $P$ is defined using the coordinates $(x, y)$.
Using the curve formula, we can recover $y$ for any given $x$.

Point encoding consists in finding $P_\beta = (\beta, y_\beta)$, with $y_\beta \ge 0$.
The shares are randomly sampled points $P_1, \dots, P_p \in E(\mathbb{F}_q)$ such that $\sum_i P_i = P_\beta$.
The share decoding requires summing the $p$ points and then extracting the first coordinate of the sum result.

\paragraph{Discussion}
However, this encoding \textit{prevents a general share compressibility} because we cannot sum two encoded secrets.
This constraint comes from the properties of the elliptic curve addition.
Consider two secrets $\beta_1, \beta_2 \in \mathbb{F}_q$ and their respective point-based encodings $P_{\beta_1}, P_{\beta_2}$, let $P_* = P_{\beta_1} + P_{\beta_2} = (x_*, y_*)$.
In the general case, we have $x_* \neq \beta_1 + \beta_2$.

Nevertheless, there is an exception to this impossible addition: if one of the shared values is 0.
By convention, we encode zero values using the point at infinity $\mathcal{O}$.
Hence, we have $P_{\beta} + P_{0} = P_{\beta} + \mathcal{O} = P_{\beta}$.
Thus, point-based encoding provides a limited compressibility: we can sum multiple function shares only if, \emph{at most} one of them is not null.

\paragraph{DPF Use case}
Anonymous broadcasters enable several parties to broadcast a message without revealing the message origin.
These systems often take the form of public bulletin board in which the writers are anonymous.
In recent years, DPFs have become valuable primitives to build efficient anonymous broadcasters like in Riposte \cite{corrigan-gibbs_riposte_2015} and Spectrum \cite{newman_spectrum_2022}.

Their DPF-based protocols have the following structure: (1) the servers initialize a large shared bulletin board, (2) each sender $i$ picks at random position $\alpha_i$ in the bulletin board, (3) each sender $i$ uses their message $m_i$ to share a secret point function $f_i$ such that $f_i(x) = m_i$ if $x=\alpha_i$, $0$ otherwise, (4) for each shared function $\share{f_i}$, the servers obtain a shared vector $\share{v_i}$ such that $\share{v_i[j]} = \share{f_i(j)}$, (5) they sum all the vectors $v_i$ with the bulletin board, and reveal the final bulletin board.
In other words, the DPF is used as an oblivious write operation on a secret-shared bulletin board.

If at most one sender picks an index, the limited compressibility is sufficient.
Spectrum and Riposte proposed solutions to avoid two senders from choosing the same index.

\paragraph{DCF Use case}
Many recent works \cite{jawalkar_orca_2024,kumar_compact_2024,ryffel_ariann_2022} considered FSS to optimize non-linear operations in privacy-preserving machine learning.
For instance, Kumar et al. \cite{kumar_compact_2024} used DCF to perform efficient and secure ReLU operations.
This last protocol does not require aggregating multiple shared secrets, so compressibility is not necessary; which makes point-based encoding perfectly acceptable.

\subsection{Exponent-based secret encoding}
This second encoding represents the secret value as an exponent: $P_\beta = \beta\cdot P$, with $P$, a generator of the curve.
The shares are randomly sampled points $P_1, \dots, P_p \in E(\mathbb{F}_q)$ such that $\sum_i P_i = P_\beta$; it requires sampling $p$ exponents $(s_1, \dots, s_p)$ such that $\sum s_i = \beta$.
The share decoding consists in summing $p$ elliptic curve points and then computing the discrete logarithm.
While this operation is hard in a general case, it is possible if we assume a bounded $\beta$ (e.g., $\beta < 10^6$).
This assumption is realistic in some applications such as private statistics.

\paragraph{Discussion}
This encoding provides compressibility because we can sum two encoded secrets: $P_{\beta_1} + P_{\beta_2} = \beta_1 \cdot P + \beta_2 \cdot P = (\beta_1 + \beta_2) \cdot P = P_{\beta_1 + \beta_2}$.
However, it has one drawback: the secret must be small, otherwise the decoding is impractical.

\paragraph{DPF Use case}
Some works \cite{boneh_lightweight_2021,mouris_plasma_2024} have used two-party DPFs to build private histograms.
In this application, a group of data owners wants to compute the histogram of a certain property (e.g., salary), but each individual wants to keep their value private.

DPF schemes provides a straightforward solution: (1) each individual uses their private value $a_i \in \{1 \dots N\}$ to share a secret point function $f_i$ with  $f_i(a_i) = 1$, (2) from each shared function $\share{f_i}$, the servers can deduce a shared vector $\share{v_i}$ such that $\share{v_i[j]} = \share{f_i(j)}$, (3) The servers aggregate the shared vectors $\share{v_i}$ and decode the resulting histogram.

This private histogram protocol requires share compressibility because each bin of the histogram usually contains the sum of multiple inputs.
Exponent-based encoding is then our only option.
The main drawback of exponent-based encoding is that the secret value must be small enough to compute a discrete logarithm, but this assumption is realistic in histogram-making because histogram values number are usually bounded.

\paragraph{DCF Use case}
Barczewski et al. \cite{barczewski_differentially_2025} recently proposed using DCF to estimate empirical cumulative distribution functions (EDCF) on sensitive data.
These statistics are essential notably to evaluate the performances of ML models or to perform some statistical tests \cite{barczewski_differentially_2025}.
As ECDF are related to histograms, they require the same share compressibility properties.
Thus, exponent-based encoding is necessary for such applications, assuming the secret is sufficiently small.
Like in private histograms, such a bound usually exists in practice.

\section{Conclusion}
Our work presented a DDH-based approach to build optimized multi-party FSS schemes from existing schemes.
In particular, it results into practical schemes with $O(\sqrt[3]{N})$ key sizes instead of $O(\sqrt{N})$.
We applied our technique both on DPF (i.e., FSS for point functions) and DCF (i.e., FSS for comparison functions).
Finally, our benchmark highlighted key size reductions up to a factor 10 compared to state-of-the-art schemes on realistic problem sizes.

\begin{credits}
  \subsubsection{\ackname}
  This work was supported by the Netherlands Organization for Scientific Research (De Nederlandse Organisatie voor Wetenschappelijk Onderzoek) under NWO:SHARE project [CS.011].
\end{credits}

\appendix
\section{Appendix: Proof of Theorem \ref{th:security}}
\label{app:proof-th:security}
\begin{proof}
  \textbf{Correctness:}
  We represent any input $x$ as $(\gamma, \delta)$ and $\alpha$ as $(\gamma_*, \delta_*)$.
  We need to study the FSS output in three cases: (1) $\gamma \neq \gamma_*$, (2) $\gamma = \gamma_*$ and $\delta \neq \delta_*$, and (3) $\gamma = \gamma_*$ and $\delta = \delta_*$.

  If  $\gamma \neq \gamma_*$, each party $i$ holds $\share{s_a}_i, \share{s_b}_i$ (obtained using their sub-DPF keys), such that $\sum \share{s_a}_i = \sum \share{s_b}_i = 0$ (because the sub-DPF is correct by assumption and outputs additive shares).
  For any $\delta$:

  \begin{align*}
    \prod_i \share{f(x)}_i & = \prod_i (h_\delta^{\share{s_a}_i} \cdot g_\delta^{\share{s_b}_i}) = h_\delta^{\sum_i \share{s_a}_i} \cdot g_\delta^{\sum_i \share{s_b}_i} = h_\delta^0 \cdot g_\delta^0 = g^0 \\
                           \Rightarrow \text{Decode}_{\DPF}&(\Eval_{\DPF}(k_1, x), \dots, \Eval_{\DPF}(k_p, x)) = \Decode_{\DDH}(g^0) =0
  \end{align*}

  If $\gamma = \gamma_*$ and $\delta \neq \delta_*$, each party $i$ holds $\share{s_a}_i, \share{s_b}_i$, such that $\sum_i \share{s_a}_i = r$ and $\sum_i \share{s_b}_i = 1$.
  The parties obtain:
  \begin{align*}
    \prod_i \share{f(x)}_i & = \prod_i (h_\delta^{\share{s_a}_i} \cdot g_\delta^{\share{s_b}_i}) = h_\delta^{\sum_i \share{s_a}_i} \cdot g_\delta^{\sum_i \share{s_b}_i} \\
                           & = g_\delta^{-r_\text{inv} \cdot r} \cdot g_\delta  =  g_\delta^{-1} \cdot g_\delta = g^0                                                \\
                           \Rightarrow \text{Decode}_{\DPF}&(\Eval_{\DPF}(k_1, x), \dots, \Eval_{\DPF}(k_p, x)) = \Decode_{\DDH}(g^0) =0
  \end{align*}

  If $\gamma = \gamma_*$ and $\delta = \delta_*$, we also have $\sum_i \share{s_a}_i = r$ and $\sum_i \share{s_b}_i = 1$, but $h_\delta$ is now equal to $(g_\delta^{-r_\text{inv}} \cdot g_\beta^{r_\text{inv}})$. It implies:
  \begin{align*}
    \prod_i \share{f(x)}_i & = \prod_i (h_\delta^{\share{s_a}_i} \cdot g_\delta^{\share{s_b}_i}) = h_\delta^{\sum_i \share{s_a}_i} \cdot g_\delta^{\sum_i \share{s_b}_i} \\
                           & = g_\delta^{-r_\text{inv} \cdot r} \cdot g_\beta^{r_\text{inv} \cdot r} \cdot g_\delta = g_\beta                                        \\
                           \Rightarrow \text{Decode}_{\DPF}&(\Eval_{\DPF}(k_1, x), \dots, \Eval_{\DPF}(k_p, x)) = \Decode_{\DDH}(g_\beta) = \beta
  \end{align*}

  \textbf{Privacy:}
  This proof will rely on a construction also used in \cite{corrigan-gibbs_riposte_2015} to prove the security of their DDH-based DPF: Seed-Homomorphic PseudoRandom Generators (SH-PRG).
  A PRG $G$ is seed-homomorphic if it satisfies the following property: for any seeds $s_1, s_2$, $G(s_1 + s_2) = G(s_1) \oplus G(s_2)$.
  In particular, Corrigan-Gibbs et al. \cite{corrigan-gibbs_riposte_2015} rely on DDH-based SH-PRG $G_{\DDH}$: given $L$ randomly sampled public parameters $(g_1, \dots, g_L)\in \mathbb{G}^L$, $G_{\DDH}(s) = (g_1^s, \dots, g_L^s)$ for any seed $s$.
  SH-PRGs naturally inherit all properties of a PRG, notably the fact that the output is computationally indistinguishable from true randomness if the seed is unknown to the adversary.

  Recall the two assumptions of Theorem \ref{th:security}: (1) the DDH is hard, and (2) the sub-DPF scheme is private and correct.
  To prove the privacy (see Definition \ref{def:privacy}), we must show that (for every set of corrupted parties $S \subseteq \{1\dots p\}$ of size $m$) there exists an efficient simulator that, for any input functions, outputs samples from a distribution that is computationally indistinguishable from the distribution of the DPF keys.

  Remind that each key $k_i$ contains the following elements: two sub-DPF keys $(k_i^{(a)},k_i^{(b)})$ and pairs of ``correction points'' $(h_1, g_1), \dots, (h_\nu, g_\nu)$

  By assumption, the sub-DPF scheme $(\Gen^{(*)}_{\DPF}, \Eval^{(*)}_{\DPF}, \Decode_+)$ is secure, so there exists an efficient simulator to simulate the sub-DPF keys $(k_i^{(a)},k_i^{(b)})$ (for all $i\in\{1\dots p\}$).

  To simulate the correction points, for each $\delta \in \{1 \dots \nu\}$, the simulator samples a random generator $\widetilde{g_\delta}\in\mathbb{G}$ and another random element $\widetilde{h_\delta}\in\mathbb{G}$.
  Note that ${g_\delta}$ is also sampled randomly during the key generation, so the distribution of ${g_\delta}$ is indistinguishable from the distribution of $\widetilde{g_\delta}$.

  Finally, we have to prove that the adversary cannot distinguish $h_\delta$ from $\widetilde{h_\delta}$.
  To prove this statement, we will consider successively two scenarios: (1) the adversary knows $(\delta_*, \beta)$ and (2) the adversary does not know the pair.

  Let us first assume that the adversary knows $(\delta_*, \beta)$.
  Note that we can rewrite the tuple $(h_1, \dots, h_{\delta_*}, \dots h_\nu)$ as $((g_1^{-1})^{r_\text{inv}}, \dots, (g_{\delta_*}^{-1}g_\beta)^{r_\text{inv}}, \dots (g_\nu^{-1})^{r_\text{inv}})$.
  As $g_1, \dots, g_\nu$ and $g_\beta$ are known to the attacker $((g_1^{-1}), \dots, (g_{\delta_*}^{-1}g_\beta), \dots (g_\nu^{-1}))$, these values can be interpreted as the public parameters of a DDH-based SH-PRG.
  The only condition on the public parameters is that they must be independently sampled generators.
  This condition is trivially satisfied for all $g_\delta$ with $\delta\neq\delta_*$, and we can observe that $g_{\delta_*}$ is independent of $g_\beta$ so $(g_{\delta_*}^{-1}g_\beta)$ is indistinguishable from random.

  Let $G_{DDH}^*$ then be a DDH-based SH-PRG with $((g_1^{-1}), \dots, (g_{\delta_*}^{-1}g_\beta), \dots (g_\nu^{-1}))$ as public parameters.
  We have:
  $$
    (h_1, \dots h_\nu)= ((g_1^{-1})^{r_\text{inv}}, \dots, (g_{\delta_*}^{-1}g_\beta)^{r_\text{inv}}, \dots (g_\nu^{-1})^{r_\text{inv}}) = G_{DDH}^*(r_\text{inv})
  $$

  In other words, $(h_1, \dots h_\nu)$ is the output of $G_{DDH}^*$ on the seed $r_\text{inv}$.
  The seed $r_\text{inv}$ is unknown to the adversary, so they cannot distinguish the real $(h_1, \dots h_\nu)$ from the randomly sampled $(\widetilde{h_1}, \dots \widetilde{h_\nu})$ outputted by the simulator \cite{corrigan-gibbs_riposte_2015}.

  Let us now assume that the adversary does not know $(\delta_*, \beta)$.
  If the adversary were able to distinguish $h_\delta$ from $\widetilde{h_\delta}$ without this knowledge, we could trivially build a distinguisher for an adversary knowing $(\delta_*, \beta)$.
  Since the $h_\delta$ can be formulated as the output of a DDH-based SH-PRG, this hypothetical distinguisher would break the DDH assumption (upon which the SH-PRG is built).
  By assumption, DDH is hard, so this adversary cannot distinguish real $h_\delta$ from the simulator output $\widetilde{h_\delta}$.
\end{proof}
\section{Appendix: Adapting Bunn et al. \cite{bunn_cnf-fss_2022} to DCF}
\label{subsec:adapt-bunn-et-al}

Bunn et al. \cite{bunn_cnf-fss_2022} presented a generic technique to build more efficient honest-majority DPF schemes from dishonest-majority schemes.
They applied their technique on two existing schemes: the trivial DPF (i.e., a secret-shared truth table) and on the PRG-based scheme from \cite{boyle_function_2015}.

While their constructions provides some of the best solutions in DPF (see Table \ref{tab:dpf}), they did not adapt them to comparison functions.
This leaves the literature with only weak DCF baselines \cite{boyle_function_2015,kumar_compact_2024}.
Thus, we propose to extend their work and build a DCF following the same intuition as their DPF.

\paragraph{Their DPF schemes}
To understand our adaptation, it is first important to understand their initial DPF.
Their main intuition is to represent a point function $f$ as the product of two point functions $f_a$ and $f_b$ (defined over small domains): $f(x) = f_a(\gamma) \times f_b(\delta)$ (with $(\gamma, \delta)$ the representation of $x$ in a $(\sqrt{N}\times\sqrt{N})$ grid).
Each function is shared using an existing DPF scheme.
However, additive secret shares cannot be multiplied without communications.

Instead of using additive secret sharing, Bunn et al. \cite{bunn_cnf-fss_2022} used replicated secret sharing: each party receives multiple DPF keys.
Under the honest-majority assumption, the shareholders can perform one offline multiplication on values shared via replicated secret sharing.

To sum up, Bunn et al. \cite{bunn_cnf-fss_2022} used replicated secret-sharing and the honest-majority assumption to build a technique reducing the key size of existing dishonest-majority schemes from $O(\sqrt[k]{N})$ to $O(\sqrt[2k]{N})$.
As reported in Table \ref{tab:dpf}, this technique leads to a PRG-based DPF with $O(\sqrt[4]{N})$ key size and an information-theoretic DPF with $O(\sqrt{N})$ key size.

\paragraph{Our adapted DCF}
While point functions are the product of two point functions, we can decompose comparison functions using three sub-functions (defined over smaller domains): one point function $f_a$ and two comparison functions $f_b, f_c$.
A comparison function $f$ can be expressed as $f(x) = f_a(\gamma) \times f_b(\delta) + f_c(\gamma)$ with
\begin{itemize}
  \item $f(x) = \beta$ if $x\le \alpha$, $0$ otherwise; with $x\in\{1\dots N\}$ and  $\alpha = \gamma_* \times \lceil\sqrt{N}\rceil + \delta_*$.
  \item $f_a(\gamma) = \beta$ if $\gamma = \gamma_*$, $0$ otherwise; with $\gamma\in\{1\dots \lceil\sqrt{N}\rceil\}$.
  \item $f_b(\delta) = 1$ if $\delta \le \delta_*$, $0$ otherwise; with $\delta\in\{1\dots \lceil\sqrt{N}\rceil\}$.
  \item $f_c(\gamma) = \beta$ if $\gamma < \gamma_*$, $0$ otherwise;with $\gamma\in\{1\dots \lceil\sqrt{N}\rceil\}$.
\end{itemize}

The multiplication of $f_a$ and $f_b$ enables to represent the row $\gamma_*$ containing a segment of $\beta$ values and a segment of 0 values.
Furthermore, $f_c$ covers all the rows before $\gamma_*$ full of $\beta$ values.

Like in \cite{bunn_cnf-fss_2022}, we rely on honest-majority and replicated secret sharing to perform the multiplication of $f_a$ and $f_b$.
Due to space limitations, we cannot include the detailed algorithms.

Similarly to \cite{bunn_cnf-fss_2022}, we can apply this technique either on the PRG-based schemes of \cite{boyle_function_2015} or on the trivial FSS schemes (i.e., sharing a truth table).
The PRG-based DCF has $O(\sqrt[4]{N} \cdot \sqrt{q^{p^m}} \cdot \binom{p-1}{m} \cdot(\lambda + \log q))$ key size, while the information-theoretic DCF has $O(\sqrt{N}\cdot \binom{p-1}{m} \cdot \log q)$ key size.

Since our adapted DCF only adds a sub-DCF (assumed to be secure) compared to the initial DPF, the security proofs can be easily adapted from \cite{bunn_cnf-fss_2022}.

%
%
%
\bibliographystyle{splncs04}
\bibliography{ref}

\end{document}